\documentclass[a4paper]{llncs}
 
\usepackage{microtype}

\usepackage{xcolor}
\usepackage{amsmath,amssymb}
\usepackage{graphicx}
\usepackage{algpseudocode}
\usepackage{algorithm}

\iftrue
\newcommand{\toappendix}[1]{ {\color{blue} Moved to appendix} } 
\fi

\title{Approximation Algorithms for Replenishment Problems with Fixed Turnover Times}

\author{Thomas Bosman\inst{1} \and Martijn van Ee\inst{1}\and Yang Jiao\inst{2}\and Alberto Marchetti-Spaccamela\inst{3}\and R. Ravi\inst{2} \and Leen Stougie\inst{1,4}}

\institute{Vrije Universiteit, Amsterdam, The Netherlands\\ \texttt{\{thomas.bosman,m.van.ee,l.stougie\}@vu.nl}
\and
Tepper School of Business, Carnegie Mellon University, Pittsburgh, PA, USA\\
  \texttt{ \{ravi,yangjiao\}@andrew.cmu.edu }
  \and
  Sapienza University of Rome, Rome, Italy\\
  \texttt{alberto@dis.uniroma1.it}
  \and
  Centrum voor Wiskunde en Informatica (CWI), Amsterdam, The Netherlands\\
  \texttt{l.stougie@cwi.nl}
}

\newcommand{\minsum}{\textsc{min-avg}}
\newcommand{\minmax}{\textsc{min-max}}
\newcommand{\rrd}{\textsc{rftt}}
\newcommand{\rpmt}{replenishment}
\newcommand{\rpmts}{replenishments}
\newcommand{\rpg}{replenishing}
\newcommand{\tr}{turnover time}
\newcommand{\trs}{turnover times}

\newcommand{\ps}{\textsc{Pinwheel Scheduling}}
\newcommand{\thrp}{3-\textsc{Partition}}
\newcommand{\tsp}{\textsc{tsp}}

\newcommand{\seq}[1]{\mathbf{#1}}

\newcommand{\ceil}[1]{\lceil #1 \rceil}

\newif\ifcomment\commentfalse
\def\commentON{\commenttrue}

\long\outer\def\bc#1\ec{{\ifcomment \sloppy  \textcolor{red}{
{ {#1}} }\fi }}

\long\outer\def\BC#1\EC{{\ifcomment \sloppy \par \textcolor{violet}{\#  \dotfill
{\textsc{#1}} \dotfill \#} \par \fi }}

\commentON

\begin{document}
\maketitle

\begin{abstract}
We introduce and study a class of optimization problems we coin replenishment problems with fixed \trs{}: a very natural model that has received little attention in the literature. Nodes with capacity for storing a certain commodity are located at various places; at each node the commodity depletes within a certain time, the \tr{}, which is constant but can vary between locations. Nodes should never run empty, and to prevent this we may schedule nodes for replenishment every day. The natural feature that makes this problem interesting is that we may schedule a replenishment (well) before a node becomes empty, but then the next replenishment will be due earlier also. This added workload needs to be balanced against the cost of routing vehicles to do the replenishments. In this paper, we focus on the aspect of minimizing routing costs. However, the framework of recurring tasks, in which the next job of a task must be done within a fixed amount of time after the previous one is much more general and gives an adequate model for many practical situations. 

Note that our problem has an infinite time horizon. However, it can be fully characterized by a compact input, containing only the location of each store and a \tr{}. 
This makes determining its computational complexity highly challenging and indeed it remains essentially unresolved. 
We study the problem for two objectives: \minsum\, minimizes the average tour length and \minmax\, minimizes the maximum tour length over all days. For \minmax\, we derive a logarithmic factor approximation for the problem on general metrics and a $6$-approximation for the problem on trees, for which we have a proof of NP-hardness. For \minsum\, we present a logarithmic approximation on general metrics, $2$-approximation for trees, and a pseudopolynomial time algorithm for the line. Many intriguing problems remain open.

\end{abstract}

\section{Introduction}
Imagine the following particular inventory-routing problem. A set of automatic vendor machines are spread over a country or a city. They have a certain turnover time: the number of days in which a full machine will be sold out. Replenishment is done by vehicles. Let us assume that \trs{} are machine dependent but not time dependent, and that it is highly undesirable to have an empty machine. However, the holding costs of the machine are negligible, so that we will always fill the machine to capacity. 
There is nothing against replenishing a machine before it has become empty, but then the next replenishment will due earlier as well. That is, the deadline of the next \rpmt{} is always within the \tr{} after the last \rpmt{}. Equivalently, in any consecutive number of days equal to the \tr{}, at least one \rpmt{} has to take place. Replenishing a machine earlier to combine it with the \rpmt{} of another machine that is due earlier may lead to cost savings.
The feature that makes this problem so special w.r.t. existing literature, is that it can be compactly modeled by only specifying for every machine its location and the \tr.   The feature is very natural but has hardly been studied in the existing literature. There are intriguing basic open complexity questions, and some highly non-trivial results.

The motivation for studying this problem comes linea recta from a business project for the replenishment of ATMs in the Netherlands, in which some of the co-authors are involved. The \rpmt{} of the ATMs of all the large banks in the Netherlands has been outsourced to a single company: Geld Service Nederland. Of course the real-life ATM replenishment problem is not as stylized as described above; the \tr{} is not strictly the same over time but subject to variability, there are restrictions on the routes for the vehicles, etc. But the feature that is least understood in the ATM-problem is exactly the problem of how to deal with the trade-off between \rpg{} an ATM earlier than its due date leading to a higher frequency of \rpmts{} and the savings on vehicle routing costs.


Formally, an instance of the problem that we study in this paper, which we baptize the \textsc{replenishment problem with fixed \trs} (\rrd), consists of a pair $(G,\mathbf{\tau})$, where $G=(V\cup \{s\},E,c)$ is a weighted graph with a designated depot vertex $s$ and weights on the edges $c: E \to \mathbb{R}_+$, and \trs{} $\mathbf{\tau} \in \mathbb{N}^{|V|}$, indicating that $v_j\in V$ should be visited at least once in every interval of $\tau_j$ days.

A solution consists, for each day $k$, of a tour $T_k$ in $G$ starting in and returning to the depot $s$ and visiting a subset of the vertices $J_k\subseteq V$. It is feasible if $v_j \in \bigcup_{k=t+1}^{t+\tau_j} J_k$, $\forall t$ and $\forall v_j\in V$. We will focus on solutions that repeat themselves after a finite amount of time, that is, in which $(T_k,...,T_{k+\ell})=(T_{k+\ell+1},...,T_{k+2\ell})$ for some $\ell$, and all $k$. Since all \trs{} are finite, this is no real restriction.

We consider two versions of \rrd. In the first version, called \minsum, the goal is to find a feasible solution that minimizes the average tour length. In the \minmax\, problem, we want to find a feasible solution that minimizes the maximum tour length over all days.

We emphasize that the particular feature of this model, that jobs or visits to clients recur and need to be done within each job-specific consecutive time interval occurs naturally in many problem settings. It allows any job of a recurring task to be done before its deadline, but then the next job of the task comes earlier and hence its deadline. This is a feature that, despite its natural applicability, has hardly been studied in the literature from a theoretical point of view. 

\paragraph*{Related work.}
As mentioned before, our problem can be seen as a special case of the \textsc{Inventory Routing Problem} (IRP) \cite{Coelhoetal2013}. Here, clients (vertices) have their own storage with a certain capacity and for each day a demand is specified. The clients pay holding cost over their inventory. However, omitting inventory cost, we can interpret our problem as such an inventory routing problem in which the demand at any given location is the same every day, leading to a very small input description of our problem consisting only of a location and a \tr{} (storage capacity divided by daily demand), which makes it incomparable to the inventory routing problem from a complexity point of view. Indeed it is unclear if the decision version of our problem is in NP or in co-NP.

Another closely related problem is the \textsc{Periodic Latency Problem}~\cite{CSW11}, which features the recurring visits requirement of \rrd. We are given recurrence length $q_i$ for each client $i$ and travel distances between clients. Client $i$ is considered \emph{served} if it is visited every $q_i$ time units. The server does not return to the depot at the end of each time unit (e.g. day), but keeps moving continuously between clients at uniform speed. Another difference between \textsc{Periodic Latency Problem} and \rrd{} is the objective function. Coene et. al.~\cite{CSW11} study two versions of the problem: one that maximizes the number of served clients by one server, and one that minimizes the number of servers needed to serve all clients. They resolve the complexity of these problems on lines, circles, stars, trees, and general metrics.

A problem that does share the compact input size and is in fact a very special case of our problem is known under the guise of \textsc{Pinwheel Scheduling}. It has been introduced to model the scheduling of a ground station to receive information from a set of satellites without data loss. In terms of our problem no more than one vertex can be replenished per day and all distances to the depot are the same; the interesting question here is if there exists a feasible schedule for \rpg{} the vertices. Formally, a set of jobs $\{1,...,n\}$ with periods $p_1,...,p_n$ is given, and the question is whether there exists a schedule $\sigma : \mathbb{N} \to \{1,..,n\}$ such that $j \in \bigcup_{k=t+1}^{t+p_j}\sigma_k$, $\forall t\geq 0$ and $\forall j$. 

\textsc{Pinwheel Scheduling} was introduced by Holte et al. \cite{Holteetal1989}, who showed that it is contained in PSPACE. The problem is in NP if the schedule $\sigma$ is restricted to one in which for each job the time between two consecutive executions remains constant throughout the schedule. In particular this holds for instances with density $\rho=\sum_j 1/p_j=1$~\cite{Holteetal1989}. They also observed that the problem is easily solvable when $\rho\leq 1$ and the periods are harmonic, i.e. $p_i$ is a divisor of $p_j$ or vice versa for all $i$ and $j$. As a corollary, every instance with $\rho\leq\frac{1}{2}$ is feasible.

Chan and Chin \cite{ChanChin1993} improved the latter result by giving an algorithm that produces a feasible schedule for \textsc{Pinwheel Scheduling} whenever $\rho\leq\frac{2}{3}$. In \cite{ChanChin1992}, they improved this factor to $\frac{7}{10}$. Later, Fishburn and Lagarias \cite{FishburnLagarias2002} showed that every instance with $\rho\leq\frac{3}{4}$ has a feasible schedule. All these papers work towards the conjecture that there is a feasible schedule if $\rho\leq\frac{5}{6}$. That this bound is tight  can be seen by the instance with $p_1=2$, $p_2=3$ and $p_3=M$, with $M$ large. This instance cannot be scheduled, but has a density of $\frac{5}{6}+\frac{1}{M}$.

The complexity of \textsc{Pinwheel Scheduling} has been open since it was introduced. It was only recently shown by Jacobs and Longo \cite{JacobsLongo2014} that there is no pseudopolynomial time algorithm solving the problem unless SAT has an exact algorithm running in expected time $n^{O(\log n\log\log n)}$, implying for example that the randomized exponential time hypothesis fails to hold \cite{Calabroetal2008,Delletal2014}. Since the latter is unlikely, one could conclude that \textsc{Pinwheel Scheduling} is not  solvable in pseudopolynomial time.  It remains open whether the problem is PSPACE-complete.

Similar to \textsc{Pinwheel Scheduling}, the \textsc{$k$-server Periodic Maintenance Problem}~\cite{MRTV89,Baruahetal1990,Eisenbrandetal2010} has $n$ jobs, each with a specified periodicity and a processing time. Each server may serve at most one job per time unit. However, job $i$ is required to be served exactly every $m_i$ days apart rather than within every $m_i$ days. The case $k=1, c_j=1$ for all $j$ is analogous to \textsc{Pinwheel Scheduling}, except for the exact periodicity constraint. For any $k \geq 1$, Mok et. al.~\cite{MRTV89} have shown it is NP-complete in the strong sense. For the special case when $m_i$ are multiples of each other or when there are at most 2 different periodicities, they have shown it is in P. It was shown  that even in the case of a single server and  $c_j=1$ for all $j$ the problem remains coNP-hard~\cite{Baruah93}.

d
Other related problems with a compact input representation include real-time scheduling of sporadic tasks \cite{Baruahetal03,BonifaciM12}, where
we are given a set of recurrent tasks. 
On a single machine, EDF (Earliest Deadline First) is optimal. However, we remark that  the complexity of deciding whether a given set of tasks  is feasible  has been open for a long time and only recently proved showing that  it is coNP-hard to decide  whether a task system is feasible on a single processor  even if the   utilization is bounded  \cite{Ekberg16}.

Another related problem is the \textsc{Bamboo Garden Trimming Problem} introduced by Gasieniec et. al.~\cite{GKLLMR17}. There are $n$ bamboos, each having a given growth rate, which may be viewed as inducing a periodicity. On each day, a robot may trim at most one bamboo back to height $0$. The goal is to minimize the maximum height of the bamboos. Gasieniec et. al. provide a $4$-approximation for the general case and a $2$-approximation for balanced growth rates.


\paragraph*{This paper.}

We investigate the computational complexity of both the \minmax\, and the \minsum\, version of \rrd. Mostly we will relate their complexity to the complexity of \textsc{Pinwheel Scheduling}. Some interesting inapproximability results follow from this relation. After that, we will start with some special cases. In Section \ref{treeapprox}, we give our most remarkable result, a constant factor approximation for \minmax\, on a tree, next to a less remarkable constant approximation for the \minsum\, version on the tree. In the same section, we show for \minsum\, that the problem can be solved to optimality in pseudopolynomial time on line metrics. Finally, in Section~\ref{genapprox}, we present logarithmic factor approximations for both problem versions on general metrics.

\section{Complexity}
In this section, we investigate the computational complexity for both objectives. Since our problem requires finding a shortest tour visiting some subset of vertices for every day, it is at least as hard as the \textsc{Traveling Salesman Problem} (\textsc{tsp}). However it is also interesting to note that the problems are at least as hard as \textsc{Pinwheel Scheduling} as well. For the \minmax{} objective there is a direct reduction showing that a factor 2 approximation is at least as hard as \ps: construct an unweighted star with the depot at the center and each leaf corresponding to a job in the pinwheel instance. This instance has value 2 only if there exists a pinwheel schedule and at least 4 otherwise. 

For the \minsum\,\rrd{} the reduction is a bit more involved, and given in the appendix. 
\begin{theorem} On series-parallel graphs, \label{thm:minsumhard} \minsum\,\rrd\, is at least as hard as \textsc{Pinwheel scheduling}. 
\end{theorem} 
We note that this hardness result is incomparable to the \textsc{tsp} reduction. Pinwheel is neither known to be NP-hard nor in NP, although it is conjectured to be PSPACE-complete. 

Lastly, as Theorem~\ref{star} shows, the \minmax{} \rrd{} remains hard even on star graphs (where \textsc{TSP} is trivial). A reduction can be found in the appendix. 
\begin{theorem}
\label{star}
\minmax\,\rrd\, is NP-hard on star graphs.
\end{theorem}

\section{Approximation on trees}
\label{treeapprox}
In this section we give a 2-approximation for \minsum\, and a 6-approximation for \minmax\, on trees. 

We start out with a simplifying result, which will also be of use in the next sections. The proof of Lemma~\ref{lem:roundingto2powers}, which is not hard to derive, can be found in the appendix.  


\begin{lemma} 
  \label{lem:roundingto2powers}
Given an instance $(G, \tau)$ of \rrd, let $\tau'$ be found by rounding every \tr\ in $\tau$ down to a power of $2$. Then $OPT(G,\tau') \leq 2OPT(G,\tau)$ for both \minsum\ and \minmax\ objectives.
\end{lemma} 

 In the remainder we assume w.l.o.g. that $G$ is rooted at $s$ and that \trs\ are increasing on any path from the depot to a leaf node in $G$. Furthermore, for an edge $e$ in $E$ we define $D(e)$ to be the set of vertices that are a descendant of $e$. We also need the following definition.  

\begin{definition}[tt-weight of an edge] 
	For any edge in $G$ we define: 
	\[q(e) = \min_{j \in D(e)} \tau_j .\] 
	We call this quantity the tt-weight (\tr-weight) of $e$. 
\end{definition} 

This definition allows us to express the lowerbound in Lemma~\ref{lem:treelb}.
\begin{lemma}[tt-weighted tree] 
	\label{lem:treelb} 
	For an instance $(G,\tau)$ of the \rrd{} on trees it holds that the average tour length is at least:
\[L(G,\tau) := 2\sum_{e\in E} \frac{c(e)}{q(e)}.\] 
\end{lemma} 
\begin{proof}
	This follows immediately from the fact that $\frac{2}{q(e)}$ lower bounds the number of times edge $e$ must be traversed on \emph{average} in any feasible solution. 
\end{proof} 
Since the maximum tour length is at least the average tour length, Lemma \ref{lem:treelb} also provides a lower bound for the \minmax\ objective.

An approximation for \minsum{} \rrd{} is thus found by rounding all \trs{} to powers of 2 and then visit each client $j$ on every day that is a multiple of $\tau_j$. Since in that case the lower bound of Lemma~\ref{lem:treelb} is exactly attained on the rounded instance, Lemma~\ref{lem:roundingto2powers} implies the following theorem.

\begin{theorem}
There is a 2-approximation for \minsum\,\rrd\, on trees. 
\end{theorem}

\subsection{MIN-MAX}
We will now show that we can achieve a $6$-approximation for \minmax\,\rrd\, on trees by providing  a $3$-approximation algorithm if all \trs{} are powers of $2$ and then applying  Lemma~\ref{lem:roundingto2powers}. 

The main idea is to take a TSP-tour and recursively split it to obtain a schedule for the clients with increasing \trs{}. During the splitting process, we assign each client $j$ on that tour to a congruence classes $\bar{a}_{\tau_j} = \{k\in \mathbb{N} | k \equiv a \pmod{\tau_j}\}$ for some $a\leq \tau_j$, to indicate we want to visit $j$ on each day in $\bar{a}_{\tau_j}$. Similarly, we distribute all edges $e$ to a congruence class $\bar{a}_{q(e)}$. We do this ensuring that on any given day, we can create a tour through all clients associated with that day, using the edges associated with that day plus a small set of extra edges.

Let us define some further notation. For a given congruence class $\bar{a}_{m} \subseteq \mathbb{N}$, we denote $g(\bar{a}_m)\subseteq V$ the set of vertices and $f(\bar{a}_m)\subseteq E$ the set of edges assigned to that class. Note that $\bar{a}_m$ and $(\overline{a+m})_m$ define the same congruence class, so $f(\bar{a}_m)=f((\overline{a+m})_m)$. Then, for any $k \in \mathbb{N}$ we have that $J_k$, the set of clients we need to visit on day $k$, is 
\[ J_k = \bigcup_{m\in \mathbb{N}, a\leq m| k\in \bar{a}_m} g(\bar{a}_m) .\]


\begin{algorithm}
  \caption{Algorithm for recursively constructing $f(\cdot)$ and $g(\cdot)$}
  \label{alg:treeschedule}
\begin{algorithmic}
\Function{RecurseTreeSchedule}{$\seq{d},a,m$}
\Require $\seq{d}$,  a connected sequence of edges in $G$,  powers of $2$ turnover times $\tau$; $a,m$,  integers 
\If{ $\seq{d} \neq \emptyset$ } 
  \State $f(\bar{a}_{m}) = \{ e\in \seq{d} \mid q(e) = m \}$
  \State $g(\bar{a}_{m}) = \{ j \in V(\seq{d}) \mid \tau_j = m\}$ 
  \State $k = \max_{k'}$ s.t. $ \sum_{i\in[k'-1] \mid q(d_i)>m} \frac{c(d_i)}{q(d_i)} \leq \frac12  \sum_{i\in [n] \mid q(d_i)>m} \frac{c(d_i)}{q(d_i)}$ 
    \State $\seq{d}^1 = (d_1,\dots,d_{k-1})$ 
    \State $\seq{d}^2 = (d_{k+1},\dots,d_n)$ 
    \State \textsc{RecurseTreeSchedule}($\seq{d}^1, a, 2m$), \textsc{RecurseTreeSchedule}($\seq{d}^2, a+m, 2m$)
  \EndIf
\EndFunction
\end{algorithmic}
\end{algorithm}

The assignment of vertices and edges to classes is guided by the recursive splitting of a TSP-tour in $G$. The full procedure for constructing $f(\cdot)$ and $g(\cdot)$ is shown in Algorithm~\ref{alg:treeschedule}. The algorithm is initially called with $\seq{d}$, a TSP-tour visiting all vertices in $G$, and $a=m=1$ and will determine the set of vertices to be visited on every day (i.e., those congruent to $\bar{a}_{1}$). Then the first (second) recursive call determines the sets of vertices with \tr{} $2$ that will be visited on odd (even) days.
Analougously,  \textsc{RecurseTreeSchedule}($\seq{d}^1, a, m$) will return the set of vertices with \tr{} $m$ to be visited on days in the congruence class  $\bar{a}_{m}$ and the two recursive calls will return the set of vertices with \tr{} $2m$ that are visited on days $a, a+2m, a+4m, \ldots$ and $a+m, a+3m, a+5m, \ldots$, respectively.

In the remainder we assume that any call to $f(\cdot)$ and $g(\cdot)$ returns the empty set for any argument that is not explicitly handled in Algorithm~\ref{alg:treeschedule}. Note that we use the notation $V(A)$ to denote the vertices incident to edges in $A\subseteq E$.
\begin{lemma}
  \label{lem:algcorrect} 
  After Algorithm~\ref{alg:treeschedule} terminates, each vertex $j$ appears in some set $g(\bar{a}_{\tau_j})$ for some $a$. 
\end{lemma}
\begin{proof}

Note that $\seq{d^1} \cap \seq{d^2}  = \emptyset$ and that 
$|\seq{d^1} \cup \seq{d^2} | = |\seq{d}| -1$; since $\seq{d}$ is  a connected set of edges then 
in each call to  \textsc{RecurseTreeSchedule}, $V(\seq{d}^1) \cup V(\seq{d^2}) = V(\seq{d})$. Therefore no vertex is skipped in the construction of $g(\cdot)$.  
\end{proof}

In order to find a tour  on  day $k$ through the vertices in $J_k$  
we use edges in $\bigcup_{h=1,2,\dots, m} f(\bar{a}_{h})$; as we already observed  this set of edges does not necessarily connect vertices in $g(\bar{a}_{m})$ to the depot. The next lemma shows that a tree that connects all vertices in $g(\bar{a}_{m})$ to the root can be found by considering
$\cup_{h=1,2,\dots, m} f(\bar{a}_{h})$
and adding a shortest path from some vertex in $g(\bar{a}_{m})$ to the depot.

\begin{lemma} 
  \label{lem:algtreespan}
  Let $a,m$ be such that $f(\bar{a}_m)$ is nonempty. Let $P$ be the set of edges on the shortest path connecting some arbitrary edge in $f(\bar{a}_m)$ to the root of $G$. Then the following edge set forms a connected component:
  \[T(\bar{a}_m) :=  P \cup (\bigcup_{h=1,2,\dots, m} f(\bar{a}_{h})). \]
  Moreover, $T(\bar{a}_m)$ spans $\bigcup_{h=1,2,\dots,m/2, m} g(\bar{a}_h)$. 
\end{lemma}  
\begin{proof}
  To prove our first claim, we first show that for $k\leq m$, $f(\bar{a}_k)$ either induces at most one connected component, or each component it induces is incident to a component induced by $\bigcup_{h=1,2,..,k/2} f(\bar{a}_h)$. Then, we will show that if $f(\bar{a}_k)$ induces at most one connected component, it is incident to $P \cup (\bigcup_{h=1,2,..,k/2} f(\bar{a}_h))$.

  Suppose $f(\bar{a}_k)$ does not induce at most one component. Note that $f(\bar{a}_k)$ is the subset of edges in some connected edge sequence $\seq{d}$ through $G$ that have tt-weight $k$. But by the way tt-weight is defined and the fact that $G$ is a tree, a simple path connecting disjoint edges with tt-weight $k$, can only consist of edges with tt-weight at most $k$. So every two components in $f(\bar{a}_k)$ are connected through a path of edges with tt-weight of at most $k$. Moreover since the sequence $\seq{d}$ used to construct $f(\bar{a}_k)$ is a subset of the sequence used to construct $f(\bar{a}_{k/2})$, by induction these connecting paths must be contained in $\bigcup_{h=1,2,\dots,k/2} f(\bar{a}_{h})$, as required. 

  Next we show that for any $k\leq m$ such that $f(\bar{a}_k)\neq \emptyset$, if $f(\bar{a}_k)$ is not incident to $P$ then it is incident to $\bigcup_{h=1,\dots,k/2}f(\bar{a}_h)$. 

  Let $\seq{d}$ be the sequence that was used to construct $f(\bar{a}_k)$. Since $\seq{d}$ contains all edges in $f(\bar{a}_m)$ and $P$ contains at least one such edge, there exists a minimal path $Q$ that contains some edge $e$ in $f(\bar{a}_k)$ such that $Q$ is connected to $P$. Moreover since $Q$ is minimal and $P$ contains the root, $e$ must be the edge furthest away from the root on $Q$. This implies that all edges on $Q$ have tt-weight $k$ or less.  Now suppose that $Q$ contains edges with tt-weight strictly less than $k$. Then those edges are necessarily in $\bigcup_{h=1,\dots,k/2} f(\bar{a}_h)$ and therefore $f(\bar{a}_k)$ is incident to that set. If not then $Q$ is strictly contained in $f(\bar{a}_k)$ and therefore $f(\bar{a}_k)$ is connected to $P$.


  The first claim of our lemma now follows by induction. $P \cup f(\bar{a}_1)$ is clearly connected. If $P \cup \bigcup_{h=1,2,\dots,k/2}f(\bar{a}_h)$ is connected, we get that $f(\bar{a}_k)$ is either empty or is connected to $P$ or to $\bigcup_{h=1,2,\dots,k/2}f(\bar{a}_h)$, and the result follows. 

  To prove our second claim, suppose that for some $k$ and $j \in g(\bar{a}_k)$ it holds that no edge incident to $j$, is in $\bigcup_{h=1,2,\dots,k} f(\bar{a}_h)$. We will show that that $j$ appears on $P$, from which our claim immediately follows.

  Let $\seq{d}$ be the sequence used to construct $g(\bar{a}_k)$. The edge $e$ incident to $j$ that is closest to the root, satisfies $q(e)\leq k$. So, it cannot be in $\seq{d}$ otherwise it would be contained in $\bigcup_{h=1,2,\dots,k} f(\bar{a}_h)$. But this implies that $e$ cuts off every edge in $\seq{d}$ from the root, and therefore $e$ appears on $P$, as claimed, concluding the proof. 
\end{proof} 

The next lemma allows us to bound the cost of edges included in $f(\bar{a}_h)$.

\begin{lemma}
  \label{lem:algtreecost}
  During each (recursive) call to \textsc{RecurseTreeSchedule}, it holds that 
  \[\sum_{e \in \seq{d}\mid q(e) \geq m} m\frac{c(e)}{q(e)} + \sum_{h=1}^{m/2} \sum_{e \in f(\bar{a}_{h}) \mid   q(e) = h  } c(e) \leq L(G,\tau) .\] 
\end{lemma}
\begin{proof}
 The proof is by induction on $m$.  Since we initially call the algorithm with $\seq{d}$ a TSP-tour in $G$, which visits each edge twice, it clearly holds for $m=1$. 

  Now for $m>1$, suppose it holds for all smaller $m$. Without loss of generality, suppose we have a call to the function with input $\seq{d}^1, a, m$, such that $\seq{d}, a, m/2$ are the input parameters for its parent in the call stack. 

   \begin{align*}
   & \sum_{e \in \seq{d}^1\mid q(e) \geq m} m\frac{c(e)}{q(e)} + \sum_{h=1,2,\dots,\frac{m}{2}} \sum_{e \in f(\bar{a}_{h})\mid   q(e) = h  } c(e) \\
 = & \sum_{e \in \seq{d}^1\mid q(e) \geq m} m\frac{c(e)}{q(e)} + \sum_{e \in f(\bar{a}_{m/2})\mid   q(e) = m/2  } \frac{m}{2} \frac{c(e)}{q(e)}  + \sum_{h=1,2,\dots,\frac{m}{4}} \sum_{e \in f(\bar{a}_{h})\mid   q(e) = h  } c(e) \\
 \leq & \sum_{e \in \seq{d}\mid q(e) \geq m} \frac{m}{2}\frac{c(e)}{q(e)} + \sum_{e \in f(\bar{a}_{m/2})\mid   q(e) = m/2  } \frac{m}{2} \frac{c(e)}{q(e)}  + \sum_{h=1,2,\dots,\frac{m}{4}} \sum_{e \in f(\bar{a}_{h})\mid   q(e) = h  } c(e) \\
 \leq & \sum_{e \in \seq{d}\mid q(e) \geq m/2} \frac{m}{2}\frac{c(e)}{q(e)} + \sum_{h=1,2,\dots,\frac{m}{4}} \sum_{e \in f(\bar{a}_{h})\mid   q(e) = h  } c(e) \leq L(G,\tau) \\
   \end{align*} 
   For the first equality, we split the second sum into an $h=m/2$ part and an $h=1,\ldots,m/4$ part.	In the first inequality we used the way $\seq{d}^1$ and $\seq{d}^2$ are determined in Algorithm~\ref{alg:treeschedule}, in the second inequality we used that $f(\bar{a}_m)\subseteq\seq{d}$ and in the last inequality we used the inductive hypothesis, concluding the proof.
\end{proof}

We are now ready for the main theorem. 
\begin{theorem}
  There is a $6$-approximation for \minmax\,\rrd\, on trees. 
\end{theorem}
\begin{proof}
  We first round all \trs{} down to powers of $2$, which loses a factor of $2$ in the optimal solution. We then use Algorithm~\ref{alg:treeschedule} to construct $f(\cdot)$ and $g(\cdot)$ thus determining the set of vertices $J_k$ to be visited on  day $k$.  By Lemma~\ref{lem:algcorrect} this defines a feasible schedule. 
  
If we then take $T(\bar{k}_{\tau_{max}})$ as in Lemma~\ref{lem:algtreespan}, we get a tree that spans $J_k$. Moreover the weight of $T(\bar{k}_{\tau_{max}})$ is at most $\frac32 OPT$: the contribution of $P$ is at most $\frac12 OPT$, since we need to reach any client at least on some day (and drive back), while the contribution of $\bigcup_{h=1,..\tau_{max}} f(\bar{k}_{h})$ is at most $L(G,\tau) \leq OPT$, which can be seen by applying Lemma~\ref{lem:algtreecost} for $m=2\tau_{max}$. Lastly, since we need a tour around $T(\bar{k}_{\tau_{max}})$, we lose another factor 2. This gives the approximation factor of $6$. 

It remains to show that Algorithm~\ref{alg:treeschedule} runs in polynomial time, and that we can find a polynomial representation for the schedule. 
For the first claim, note that in each recursive call to the algorithm, the following equality holds $|\seq{d^1} \cup \seq{d^2}| = |\seq{d}| -1$; hence the algorithm terminates after at most $2|E(G)|$ calls. 

For the second claim, the crucial observation is that we only need to store the entries of $g$ for $a$ and $m$ such that  $g(\bar{a}_m)$ is nonempty. Since at most one entry is defined in every call to the algorithm, and we can simply check if $k\equiv a \pmod{m}$ for all stored entries, the claim, and the theorem, follow.
\end{proof}

\subsection{MIN-AVG on the line}
As an even more special underlying metric, we might consider the \minsum\, problem on the line (on a path). For the \minmax\ version this case is trivial, but for the \minsum\ version its complexity is unclear: we do not know whether it is in NP, although we expect it to be NP-hard.

On the positive side we can show that the problem is not strongly NP-hard.
\begin{theorem}
\label{thm:line}
\minsum\ on the line can be solved in pseudopolynomial time.
\end{theorem}
The proof of Theorem~\ref{thm:line} is deferred to the appendix, where we give a DP that finds an optimal schedule in polynomial time for any instance with polynomially bounded \trs{}.

\section{Approximation on general graphs}
\label{genapprox}
We will now present logarithmic approximations for both objectives. Note that an $O(\log \tau_{max})$-approximation is readily achieved; simply treat the sets of clients with equal \tr{} as independent instances. For \minsum{}, the problem with equal \trs{} is simply \tsp, for the \minmax{} we get a problem sometimes called the $k$-\tsp{}, for which a $\frac52$ approximation is known~\cite{Fredericksonetal1976}. Since by rounding to powers of 2, we ensure there are $O(\log(\tau_{max})$ different \trs{}, we get Theorem~\ref{thm:genapprox1} (a formal proof can be found in the appendix).

\begin{theorem}
  \label{thm:genapprox1}
  \minmax{} and \minsum{} \rrd{} have an $O(\log \tau_{max})$-approximation. 
\end{theorem}

In the case of \minmax{} it is relatively simple to adapt this idea for an $O(\log n)$-approximation by appropriately reassigning clients to lower \trs{}, as per Theorem~\ref{thm:genapprox2}. 
\begin{theorem}
\label{thm:genapprox2}
  \minmax{} \rrd{} has an $O(\log n)$-approximation. 
\end{theorem}
\begin{proof}
  We start by assuming that every \tr{} is a power of $2$. Next, we split up the instance into two new instances. To this end we first define a \tr{} $k$ to be \emph{saturated} if $|\{j \in V | \tau_j = k\}| \geq k$. In the first instance we retain the set of vertices $V_1$ with saturated \trs{}, and in the second all vertices $V_2$ with unsaturated \trs{}. Now if all \trs{} are saturated, then $\tau_{max} = O(n)$ and we can find a $O(\log n)$-approximation using Theorem~\ref{thm:genapprox1}. So what remains is to find a $O(\log n)$-approximation for the second instance.

  Since no \tr{} is saturated, it is easy to see that we can partition the vertices in $V_2$ into $\lceil \log n \rceil$ sets $W_1, W_2, W_4, \dots, W_{2^{\lceil \log n \rceil}}$, such that $|W_i| \leq i$, and such that $\tau_j \geq i$ for all $j \in W_i$. For example we could first add all vertices $j$ with $\tau_j = i$ to $W_i$ for $i \leq \lceil \log n \rceil$, and then arbitrarily distribute vertices $j$ with $\tau_j > \lceil \log n \rceil$ among the sets that have space. We now produce a schedule by visiting all clients in any set $W_k$ on different days. This is feasible and implies that at most $\log n$ clients are visited on a given day, which leads to $O(\log n)$-approximation factor, as required. 
\end{proof} 


The approach of Theorem~\ref{thm:genapprox1} does not trivially extend to the \minsum{} case. However, we may combine our result on trees with the FRT tree embeddings~\cite{FRT04}, to get a randomized $O(\log n)$-approximation. 

A more direct, and deterministic $O(\log n)$-approximation is possible as well. In particular, we use the simple heuristic of visiting each client on every day that is a multiple of its \tr{}, when \trs{} are powers of 2. We call such a schedule a \emph{synchronized} solution, and show that gives a logarithmic approximation. 

The proof of this approximation factor, which is not trivial, works by show that a synchronized schedule is no more costly than a \emph{non-decreasing} schedule, in which all tours are routed along a tree with \trs{} non-decreasing from the root. We then show how to transform any schedule to a non-decreasing one, losing a logarithmic factor in the process. As a byproduct we show that the analysis is tight, an that a non-decreasing schedule must be $\Omega(\log n)$ times more costly than OPT in the worst case. The proof of Theorem~\ref{thm:genapproxavg} can be found in the appendix.
\begin{theorem}
  \label{thm:genapproxavg}
\minsum{} \rrd{} has an $O(\log n)$-approximation. 
\end{theorem}




It is an open question whether there exists a constant factor approximation algorithm for the general case. We observe that  the approach of  first finding a    tree spanning all vertices and then using the algorithm of Section  \ref{treeapprox} is unsuccessful. In fact  there exist instances of the problem  on  a graph $G$ with $n$ vertices, such that  if we limit our attention  to  tours that for each day  use only edges of a spanning  tree of $G$  then  the obtained solution  is $\Omega( \log n)$ approximated. This implies that we need some new ideas, in order to improve the  $O( \log n)$ approximation of the previous theorem. 

\section{Conclusion}
In this paper, we considered replenishment problems with fixed turnover times, a natural inventory-routing problem that has not been studied before. We formally defined the  \rrd\, problem and considered the objectives \minsum\, and \minmax. For the \minsum\,\rrd, we showed that it is at least as hard as the intractable \textsc{Pinwheel Scheduling Problem} on series-parallel graphs and we gave a 2-approximation for trees. For the \minmax\, objective we showed NP-hardness on stars and gave a 6-approximation for tree metrics. We also presented a DP that solved the \minsum\,\rrd\, in pseudopolynomial time on line graphs. Finally, we gave a $O(\log n)$-approximation for the \minmax{} objective on general metrics.


The results that we present should be considered as a first step  in this area and many problems remain open. An intriguing open problem is the complexity of the of \rrd\ on a tree. Namely, for \minsum\ variant we conjecture that the problem is hard, and we ask whether the simple 2-approximation we provide can be improved. For the \minmax\  variant it is open whether the problem is APX-hard and whether we can improve the 6-approximation,   

Next to replenishing locations with routing aspects as we studied in this paper, scheduling problems modeling maintenance or security control of systems, form a class of problems to which this model naturally applies. It would be interesting to study such fixed turnover time problems in combination with scheduling. Would this combination allow for more definitive results?



\bibliographystyle{plain}
\bibliography{replenishmentlit}
\pagebreak
\appendix

\section{Proof of Theorem \ref{thm:minsumhard}}
\begin{proof}
Given an instance $p_1,...,p_n$ of \ps, create an instance $(G,\tau)$ of \minsum{}\,\rrd{}. We define \[V = \{s, w^1, w, w^2\} \cup V^1 \cup V^2,\] where $V^i = \{v_1^i,...,v_n^i\}$ for $i =1,2$, \[E = \{(s,v) | v\in V^1 \cup V^2\} \cup \{(w^i, v_j^i) | \forall j; i=1,2\} \cup \{(w^1, w), (w^2, w)\},\] and $\tau_{w^1}=\tau_{w^2}=\tau_{w}=1$, $\tau_{v_j^1}=\tau_{v_j^2}=p_j$. All edge weights are 1. See Figure \ref{fig:minsumreduction} for an illustration. 

We claim that the instance $(G,\tau)$ has a solution of cost $6$ if and only if $p_1,...,p_n$ is a feasible \textsc{Pinwheel Scheduling} instance. 

For the `if'-direction, suppose we have a feasible pinwheel schedule. Then we create a replenishment schedule as follows: we take the set of jobs visited on day $k$, $J_k = \{v_j^1,w^1,w,w^2,v_j^2\}$, where job $j$ is the job scheduled on day $k$ in the pinwheel solution (when no job is visited, pick one arbitrarily). The pinwheel schedule then guarantees that the periods of jobs in $V^1\cup V^2$ are satisfied, while jobs $w^1,w,w^2$ are visited every day, as required. Now since any tour $(s,v_j^1,w^1,w,w^2,v_j^2,s)$ has length 6, we can do this within the claimed cost. 

For the `only if'-direction, note that any replenishment schedule must have cost $6$ at least, since no tour that visits $w^1, w, w^2$ costs less than that. Moreover, any tour visiting those three vertices that is not of the form $(s,v_j^1,w^1,w,w^2,v_j^2,s)$, will cost strictly more than $6$. It follows that if the cost of the replenishment schedule is $6$, every tour visits at most one job from $V^1$ (and the same for $V^2$). Since $\tau_{v_j^1}=p_{j}$ for all $j$, this directly implies that the \ps{} instance is feasible. 
\end{proof}
\begin{figure}[!h]
\centering
\caption{Instance created in the proof of Theorem \ref{thm:minsumhard}.}
\label{fig:minsumreduction}
\includegraphics[width=.3\textwidth]{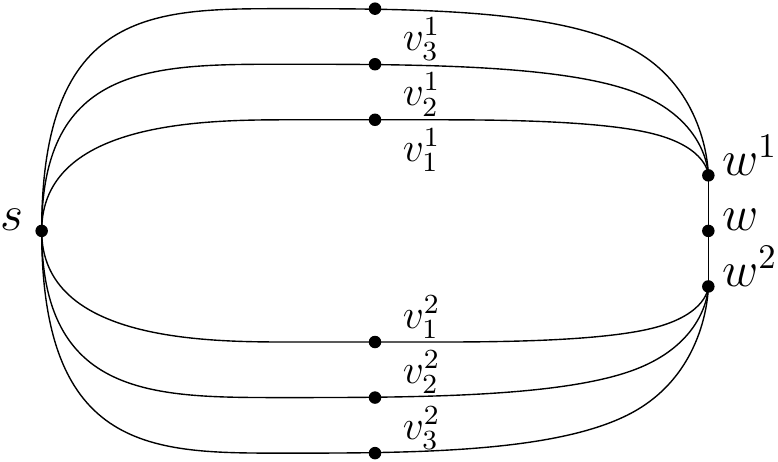}
\end{figure}

\section{Proof of Theorem \ref{star}}
In \thrp, we are given $3m$ integers $a_1,\ldots,a_{3m}$ and an integer $B=\frac{1}{m}\sum_i a_i$ such that $\frac{B}{4}<a_i<\frac{B}{2}$ for all $i$. The question is whether we can partition the integers into $m$ sets of three integers that add up to $B$ \cite{ga1979computers}.

\begin{proof}
  Given an instance of \thrp{}, create a weighted star graph $G$ with the depot at the center and for every integer $a_i$ a leaf vertex attached with an edge of weight $a_i/2$. Finally set the \tr{} to $m$ for every leaf. We will show that the \rrd{} instance has value $B$ if and only if we have a YES-instance for \thrp{}. 

  Given a valid partition of the integers, clearly one can assign every set of 3 integers a unique day in $1,\dots,m$ and visit the associated leaf on every multiple of that day for a valid \rrd{} solution. That the opposite direction works as well hinges on the fact that $\frac{B}4 < a_i$ for all $i$, so we cannot visit more than 3 clients on one day. Since after $m$ days all clients must have been visited due to the \trs{}, it follows that the first $m$ days of the schedule corresponds to a valid partition.
\end{proof}

\section{Proof of Lemma~\ref{lem:roundingto2powers}}

\begin{proof}

Let $(G,\bar{\tau})$ denote the instance found from $(G,\tau)$ by rounding every \tr\ up to a power of $2$. Since any schedule remains feasible if we round up the \trs, we have that $OPT(G,\bar{\tau}) \leq OPT(G,\tau) \leq OPT(G,\tau')$.

Suppose we have an optimal solution for $(G,\bar{\tau})$ in which $\bar{T}_k$ is scheduled on day $k$. We can construct a feasible schedule for $(G,\tau')$ by scheduling the concatenation of $\bar{T}_{2k-1}$ and $\bar{T}_{2k}$  on day $k$. The maximum tour length in this schedule is at most twice that of the optimal solution for $(G,\bar{\tau})$ and every tour from the original schedule is visited exactly twice in the new schedule, so this yields a factor 2 increase in both the \minmax\ and the \minsum\ objective. 
\end{proof}

\section{A dynamic program that proves Theorem~\ref{thm:line}}
In this section we will show how to solve the \minsum{} problem on the line in pseudopolynomial time. Since we are minimizing the average, it is easy to see that we can reduce this problem to two times the \minsum\, problem on the half-line (a path with the depot in one of the leaves). On the half-line each vertex $i$ has a distance $d_i\in \mathbb{N}$ from the origin. Suppose vertices are numbered such that $d_1\leq\ldots\leq d_n$. We present a pseudopolynomial time dynamic programming agorithm for this problem, based on the following observations.

First of all, we note that on any tour visiting vertex $j$ automatically visits every vertex $i<j$. As in the tree case, we therefore assume that $\tau_i\leq \tau_j$ for $i<j$. Thus, after visiting $j$, all  $i\leq j$ have a remaining \tr{} of $\tau_i$. For the dynamic program to work, we guess $L$, the day on which vertex $n$ is visited for the first time and try all guesses between $1$ and $\tau_n$.

The dynamic program now works as follows. Suppose we are given the optimal solution for vertices $1,\ldots,i-1$ when only considering the days $1,\ldots,k$. Now we want to include $i$ in the optimal solution for the first $k$ days. If $k<\min\{ \tau_i,L\} $, it is not necessary to visit $i$ during the first $k$ days, and hence it is optimal to take the optimal solution for the first $i-1$ vertices and $k$ days. Otherwise, we need to visit $i$ on some day $\ell$ in $\{ 1,\ldots, \min \{ \tau_i,L\} \} $. Before day $\ell$, we only need to visit the vertices $1,\dots,i-1$. Thus, we take the optimal $\ell-1$ tours for visiting the first $i-1$ vertices in the first $\ell-1$ days. After day $\ell$, all vertices have the same remaining \tr\ as they had at time zero. Hence, we can take the optimal tours for the first $i$ vertices and $k-\ell$ days.

Let $\phi_L(i, k) := \text{the minimum cost of the first $k$ tours visiting vertices $1,...,i$}.$
We initialize $\phi_L(0, k)=\phi_L(i, 0)=0$ and we  use the recursion: 
\[ \phi_L(i,k) = \begin{cases} \phi_L(i-1, k) &,\text{if } k<\min\{\tau_i,L\} \\ 
\min_{\ell=1,...,\min \{ \tau_i,L\}}\phi_L(i-1,\ell-1) + d_i + \phi_L(i,k-\ell) &,\text{else} \end{cases}\] 

The optimal solution is the schedule that corresponds to the value $L\in\{1,\ldots,\tau_n\}$ minimizing $\phi_L(n,L)/L$. Note that the algorithm runs in time $O(n\tau_n^3)$, implying the following.

\begin{theorem}
\minsum\ on the line can be solved in pseudopolynomial time.
\end{theorem}

\section{Proof of Theorem~\ref{thm:genapprox1}} 
\begin{proof}
By Lemma~\ref{lem:roundingto2powers} we may assume every $\tau_i$ is a power of $2$ so that there are at most $\log \tau_{max}$ different \trs{}. We simply treat the sets of vertices with the same \tr{} as separate instances and concatenate the solutions. Our result then follows from the fact that for all these instances a constant factor approximation is available. In the case of the \minmax{} objective we get the $k$-\tsp{} problem, where $k$ is equal to the \tr{} of the vertices in the instance. In the case of \minsum{}, we need to minimize the sum over all $k$ tours. But since all \trs{} are equal there is no advantage to visiting vertices on different days, hence we recover a simple \tsp{} problem. 
\end{proof}
\section{Proof of Theorem~\ref{thm:genapproxavg} and tightness of analysis}

This section provides two proofs of Theorem~\ref{thm:genapproxavg}, which shows that we can get a $O(\log n)$-approximation for the \minsum\ objective as well. The first proof is a direct application of metric tree embeddings. 

\begin{proof}[Proof of Theorem~\ref{thm:genapproxavg}]
\sloppy We will apply the FRT tree embedding~\cite{FRT04} of the initial instance and then use the $2$-approximation for tree metrics to obtain the final solution. Given the instance $(G,\tau)$, let $T$ be a random tree produced by the tree metric approximation with $O(\log n)$ distortion. Then $d_G(u,v) \leq d_T(u,v)$ and $E[d_T(u,v)] \leq O(\log n) d_G(u,v)$. Let $S$ be the solution produced by the $2$-approximation for \minsum{} \rrd{} on the tree metric $T$. Then $E[S] \leq 2 E[OPT(T,\tau)] \leq O(\log n) E[OPT(G,\tau)]$ by linearity of expectation on the sum over the edges.
\end{proof}

The second proof arises from a more natural algorithm given by a simple greedy strategy. We round all periods to powers of 2 and delay visiting any client for as long as possible. Next, for every day we use any constant factor approximation for TSP to calculate a tour on the clients whose visit can no longer be delayed. We call this a \emph{synchronized} solution.  It takes some work to show this does indeed provide a logarithmic approximation though. We do this by showing that any synchronized solution is no more costly than a \emph{non-decreasing} solution, in which every tour is based on a tree that has the clients ordered by ascending \trs{} from root to leaves. We then show that such a non-decreasing solution costs at most $O(\log n)$ times the optimal solution , and provide an example showing this analysis is tight.  Moreover, we show that the optimal non-decreasing solution is at most twice as costly as the optimal synchronized solution. This implies that that any sub-logarithmic approximation algorithm must avoid finding such solutions.

As always we will assume that \trs{} are rounded to powers of $2$. Let us define a \emph{synchronized} solution, as one where a client with \tr{} $2^i$ is visited on each day that is a  multiple of $2^i$, for all $i$. We define a \emph{non-decreasing} tree as a tree on the depot and subset of clients, such that the \trs{} on every path from the depot to the leafs are non-decreasing. A non-decreasing solution is a solution in which for each day the tour is given by visiting the clients of a non-decreasing tree in depth first order. 

The following two lemmas show that optimal synchronized and non-decreasing solutions differ in cost by at most a constant factor. 

\begin{lemma}
  \label{lem:syncleqnondecr}
  The optimal synchronized solution costs at most two times the optimal non-decreasing solution. 
\end{lemma}
\begin{proof}
  Suppose we have a non-decreasing solution. Let $T_i$ be the non-decreasing tree associated with day $i$, for $i \in \mathbb{N}$. We will show that we can find a set of trees $T_i'$ for $i\in \mathbb{N}$ that cost at most as much as $T_i$ on average, and such that any client $v$ appears in tree $T_i$ if $i$ is a multiple of $\tau_{v}$.  A synchronized solution can then be found by taking the tour on day $i$ to be a tour depth first search in $T_i'$, losing a factor $2$. 
  
  Iteratively, from $j=0,\ldots,\log \tau_{max}$, we build the new trees. In iteration $j$ and for all $i$, select all unmarked edges in $T_i$ that are used on a path from the depot to a client $v$ with $\tau_{v}=2^j$, and mark them. Then insert these edges in the tree $T'_k$ for the earliest following day $k$ that is a multiple of $2^j$, so $k=\ceil{\frac{i}{2^j} } 2^j$.   
  
  We now show by induction that after iteration $j$, for each day $k$ that is a multiple of $2^j$, $T_k$ is a tree connecting the depot to all clients $v$ with \tr{} $2^j$. The base case $j=0$ follows from the fact that the trees $T_i$ are non-decreasing and must contain every client with \tr{} 1. For higher $j$, it is easy to see that $T_k$ must contain a path from $v$ to $w$ for all clients $v$ with $\tau_v=2^j$ and some $w$ with $\tau_w< 2^j$. But $T_k$ already contains $w$ by our inductive hypothesis and the result follows. 
\end{proof}

We remark that we have corresponding converse result, as per Lemma~\ref{lem:nondecrleqsync}. 
\begin{lemma}
  \label{lem:nondecrleqsync}
 The optimal non-decreasing solutions costs at most two times the optimal synchronized solution. 
\end{lemma}
\begin{proof}
  We may assume that any synchronized solution uses at most $\log(\tau_{max})+1$ distinct tours, lets label them $T_0,\ldots,T_{\log(\tau_{max})}$, where $T_j$ visits all clients with \tr{} at most $2^j$. Furthermore define $\Delta_0 = c(T_0)$ and $\Delta_j = c(T_{j}) - c(T_{j-1})$, for $j=1,\ldots, \log(\tau_{max})$. Then it holds that the cost of the synchronized solution is 
  \[ c_{sync} = \sum_{j=0}^{\log(\tau_{max})} \frac1{2^j} \Delta_j .\] 

  Now note that we can create a non-decreasing tree for any day with $2^j$ as its largest power of $2$ factor from the synchronized solution, by taking the union of $T_0,\ldots, T_j$ where we shortcut every client $v$ with $\tau_v < 2^j$ in tour $T_j$. But the cost of such a solution is
 \begin{align*} 
 c_{ndecr} &=  \sum_{j=0}^{\log(\tau_{max})} \frac1{2^j} c(T_j) \\
 &= \sum_{j=0}^{\log\tau_{max}} \frac1{2^j} \sum_{i=0}^j \Delta_i  \\
 &= \sum_{i=0}^{\log\tau_{max}}  \sum_{j=i}^{\log\tau_{max} } \frac1{2^j}  \Delta_i\\
 &\leq 2  \sum_{i=0}^{\log\tau_{max}} \frac1{2^i} \Delta_i = 2c_{sync} 
 \end{align*} 
\end{proof}

Our main result follows from showing that we can always find a non-decreasing solution of cost $O(\log n)$ times $OPT$. We use the following tree pairing Lemma by Klein and Ravi~\cite{klein1995nearly}. 
\begin{lemma}
  \label{lem:treepairing}
Given any tree $T$ and an even subset $S$ of its vertices, there is a pairing of vertices covering $S$ such that the tree-path induced by the pairs are edge-disjoint.
\end{lemma}

Using the tree pairing Lemma~\ref{lem:treepairing}, we will construct non-decreasing trees to approximate arbitrary trees. First we define the notations needed for the algorithm.

A \emph{non-decreasing} arc $a(\{u,v\})$ of $\{u,v\}$ is the arc between $u$ and $v$ that points from the client with lower \tr{} to the one with higher \tr{} (ties are broken arbitrarily). The client with the lower (higher) \tr{} is denoted by $L(\{u,v\})$ ($H(\{u,v\})$). We denote by $U$ the unpaired clients and $A$ the arcs of the non-decreasing tree being constructed, and require that all arcs must eventually point away from the depot.

\begin{algorithm}
  \caption{Algorithm to create non-decreasing tree from arbitrary tree} 
  \label{alg:pairing}
\begin{algorithmic}[1]
\State Initialize $U \leftarrow V$ and $A \leftarrow \emptyset$.
\While{$|U| > 0$}:
 \State Find an edge-disjoint pairing $P$ of a largest even subset of $U$.
 \For{$\{u,v\} \in P$}:
  \State $A \leftarrow A \cup a(\{u,v\})$.
  \State $U \leftarrow U \setminus H(\{u,v\})$.
 \EndFor
\EndWhile
\end{algorithmic}
\end{algorithm}

\begin{lemma}
\label{lem:nondecr}
Given an arbitrary tree $T$ of cost $c(T)$, there is a non-decreasing tree of cost at most $\lceil \log n \rceil c(T)$.
\end{lemma}

\begin{proof}
Let $T$ be a tree. We will construct a non-decreasing tree $T'$ by iteratively pairing off the vertices and directing each pair in a non-decreasing manner.

In the algorithm, we apply the pairing procedure $\lceil \log n \rceil$ times to get a non-decreasing tree of the desired cost. In each round, we pair a largest subset of $V$ such that the pairs induce edge disjoint paths in $T$. Then we direct each pair $\{u,v\}$ in ascending order of \trs{} and delete the client with higher \tr{} from consideration. These arcs are added to the arc set of $T'$. We can think of each pair as a connected component represented by the client with the smallest \tr{}. In the end, $T'$ is finalized when no unpaired vertices remain. Note that picking the vertex of minimum \tr{} as the representative per connected component ensures that the final tree is indeed directed away from the depot.

In each round, we used each edge of $T$ at most once since all pair-induced paths were edge-disjoint. Let $\kappa(t)$ be the number of vertices at the beginning of round $t$. Since each round paired off either all vertices or all but one vertex, we have $\kappa(t) = \lceil \kappa(t-1)/2 \rceil$. So the total number of rounds is $\lceil \log n \rceil$. Hence $c(T') \leq \lceil \log n \rceil c(T)$.
\end{proof}

\begin{proof}[Proof of Theorem~\ref{thm:genapproxavg}]
  Given an optimal solution, let $T_i$ be the minimum Steiner tree on the set of clients visited on day $i$, which costs no more than the tour of that day. Using Lemma~\ref{lem:nondecr} we can find non-decreasing trees $T_i'$ of cost at most $O(\log n) c(T_i)$. Turning the trees into tours loses only a constant factor, which gives us a non-decreasing solution of cost $O(\log n)$ times $OPT$. By Lemma~\ref{lem:syncleqnondecr} we may then turn this solution into a synchronized one as required, concluding the proof.
\end{proof}

The bound in the proof of Theorem~\ref{thm:genapproxavg} is tight; as there exists a class of instances where requiring a solution to be non-decreasing introduces a logarithmic optimality gap. Together with Lemma~\ref{lem:nondecrleqsync}, this implies that our algorithm does no better than $O(\log n)$ as well.   

\begin{proposition}
  \label{prop:decrtreetight}
There exists a class of instances in which there is a logarithmic optimality gap between the optimal and the optimal non-decreasing solution.
\end{proposition} 
To show that the bound in the proof of Theorem~\ref{thm:genapproxavg} is tight, we first show that Lemma~\ref{lem:nondecr} is tight. Consider the following sequence of sequences $(a^0, a^1, \ldots)$ where $a^0 = (1)$ and $a^{i+1}$ is generated by alternatingly taking an element from $a^i$ and then from the sequence $b^i = (2^i+1, 2^i+2, \ldots, 2^{i+1})$. For example:
\begin{itemize} 
  \item $a^0=1$
  \item $a^1=1,2$
  \item $a^2=1,3,2,4$
  \item $a^3=1,5,3,6,2,7,4,8$
  \item $\cdots$
\end{itemize} 

Then define the (unweighted) graph $G_i$ as the path graph with $2^i$ vertices, where the $j$th vertex has \tr{} $2^{a^i_j-1}$. See Figure~\ref{fig:figg2} for an example. 
\begin{figure}[h]
  \caption{Illustration of $G_2$, (\trs{} in circles)}
  \label{fig:figg2}
  \centering
  \includegraphics[width=0.4\textwidth]{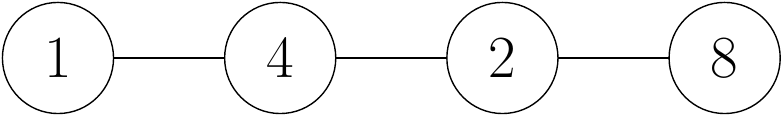}
\end{figure}

The minimum spanning tree in $G_i$ costs $2^i-1$. It is easy to check that the decreasing spanning tree produced by Algorithm~\ref{alg:pairing} costs $i2^{i-1}$. Moreover, since the solution produced attaches every vertex to a nearest vertex with lower \tr{}, it must be optimal.

To show tightness of our main theorem, we will define another class of graphs $H_i$ for $i\geq1$ that are constructed from $\{G_i\}$. The idea is to make $\tau_j$ copies of each terminal $j$, and then connect them in a regular way, for example like in Figure~\ref{fig:fig1}. 

\begin{figure}[h]
  \caption{Illustration of $H_2$}
  \label{fig:fig1}
  \centering
  \includegraphics[width=0.4\textwidth]{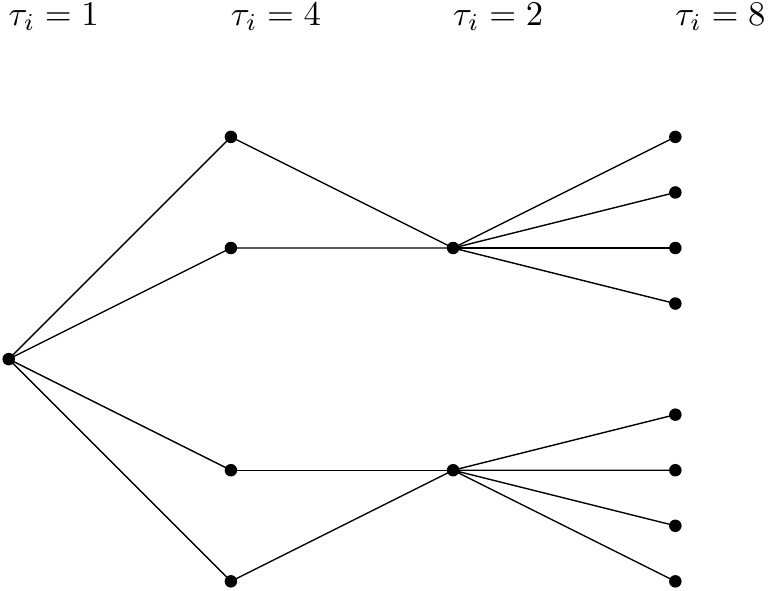}
\end{figure}

Formally $H_i$ is constructed as follows. For simplicity of description, we assume that $G_i$ is planarly embedded from left to right, and we assume that we keep a planar embedding of $H_i$ during construction.  

We first copy node $1$ to $H_i$. Now we work from left to right, starting from the second node. When we are at node $j$, we put $\tau_j$ copies of $j$ vertically above each other and to the right of the copies of $j-1$ in $H_i$. Then we connect them to the copies of $j-1$ in $H_i$ such that the graph remains planar and all copies of $j-1$ have the same degree, and all copies of $j$ have the same degree. This can be done in only one way. Furthermore we identify vertex 1 with the depot. 

\begin{proposition}
  The instance induced by $H_i$ has a logarithmic optimality gap between the optimal and the optimal non-decreasing solution.
\end{proposition}
\begin{proof}
  There exists an obvious solution that visits exactly one client of each turnover time per day, that costs $2(2^i-1)$. 

  Now suppose we impose non-decreasing constraints. In this case we need to use (on average) at least one edge pointing from a client with a lower \tr{} to one with a higher \tr{} per day. But from our reasoning on the decreasing minimum spanning tree in $G_i$, we find that the cheapest set of edges that contains at least one edge pointing from a client with \tr{} $2^i$ to one with lower \tr{} for all $i$, costs at least $i2^{i-1}$. Therefore the optimal solution under non-decreasing constraints is at least a logarithmic factor more expensive than the optimal solution. 
\end{proof}

\end{document}